\documentclass[11pt,a4paper]{article}

\usepackage[T1]{fontenc}
\usepackage[utf8]{inputenc}
\usepackage{lmodern}
\usepackage[margin=2.5cm]{geometry}
\usepackage{amsmath,amssymb,amsthm}
\usepackage{booktabs}
\usepackage{graphicx}
\usepackage{array}
\usepackage{microtype}
\usepackage[numbers,sort&compress]{natbib}
\usepackage[dvipsnames]{xcolor}
\usepackage[colorlinks=true,linkcolor=blue!50!black,citecolor=blue!50!black,urlcolor=blue!50!black]{hyperref}

\newtheorem{proposition}{Proposition}
\newtheorem{definition}{Definition}
\theoremstyle{remark}

\newcommand{\CoD}{\mathrm{CoD}}

\title{Cost of Delay for Post-Quantum Migration:\\
Putting Classical and Harvest-Now-Decrypt-Later Risk\\on One Ordered List}
\author{Animesh Shaw\\[2pt]\small Independent researcher\\\small \texttt{animesh15b@iimk.edu.in}}
\date{Manuscript draft, October 2026}

\begin{document}
\maketitle

\begin{abstract}
Organisations that must decide which assets to migrate to post-quantum cryptography first, and how that work competes with a backlog of classical security findings, have no common unit in which to compare the two. Existing post-quantum scoring models are dimensionless and quantum-only, while classical risk is expressed as annualised loss. We express both threats as a \emph{cost of delay} in currency per year on the same asset. The classical term is annualised loss expectancy. The quantum term is the rate at which deferring migration commits irreversible harvest-now-decrypt-later (HNDL) loss, computed from the required confidentiality duration, the share of traffic an adversary can record, and data turnover, using an arrival law for a cryptographically relevant quantum computer (CRQC) fitted in closed form to published expert-survey anchors. The two terms are coupled by a competing-risks factor: data that a classical attacker is likely to take first leaves less for a quantum adversary to add. We prove twelve propositions about the construction, including a dominance threshold $\lambda^*$ for the classical hazard below which the quantum term exceeds the classical term, independent of asset value; that the construction is not a weighted sum of a classical and a quantum score; and that ranking by cost of delay is optimal for sequencing work under constant loss rates. We verify the implementation against a Monte Carlo simulation, adaptive quadrature and an independent JavaScript implementation. We apply the method to four public systems whose parameters are assumed from public documentation, not elicited from practitioners: a national identity platform (MOSIP), a core-banking platform (Apache Fineract), a medical-record system (OpenMRS) and a demonstration web shop. The quantum term reorders the asset list by more than input uncertainty does for one of the four systems (signal-to-noise 1.36 against 0.14--0.46); it moves specific long-lived, quiet, recordable assets decisively (registration packets in the identity platform enter the funded top two in 55\% of 4,000 perturbed registers against 0\% under classical-only ranking) and changes only magnitudes for medium-lifetime data. Variance decomposition shows that asset value and annual turnover explain 73--91\% of the uncertainty in the quantum term while the CRQC arrival date explains 3--17\%. We report where each ordering changes under the arrival date and the discount rate rather than a single calibrated score. The case studies are illustrative; no outcome data validate the scores.
\end{abstract}


\section{Introduction}

A cryptographically relevant quantum computer (CRQC) would break the public-key schemes that establish most of today's secure sessions. Because an adversary can record ciphertext now and decrypt it once such a machine exists, data with a long confidentiality requirement is exposed from the moment it is transmitted, not from the moment the machine arrives \citep{mosca2018}. NIST has standardised a post-quantum key-establishment mechanism \citep{fips203} and has proposed deprecating quantum-vulnerable algorithms after 2030 and disallowing them after 2035 \citep{nistir8547}.

Few organisations can migrate everything at once. They hold a backlog of classical security findings ranked by annualised loss, and a separate post-quantum inventory ranked by a dimensionless score. A decision-maker then has no principled way to say whether migrating a given asset is more or less urgent than closing a classical finding on it, or on another asset.

This paper addresses that gap. The idea is simple: for each asset, ask what each further year of delay costs, separately for the classical threat and for the quantum threat, in the same unit (currency per year), and add the two after accounting for the fact that they compete for the same data.

\paragraph{Contributions.}
\begin{enumerate}
\item A \textbf{cost-of-delay formulation} in which classical and HNDL risk on the same asset are both expressed in currency per year, with no weighting or normalisation constant (Section~\ref{sec:model}).
\item A \textbf{competing-risks coupling} between the quantum term and the same asset's classical disclosure hazard, with a proof that it removes double counting (Proposition~\ref{prop:race}) and that the result is not a weighted sum of a classical and a quantum score (Proposition~\ref{prop:notsum}).
\item A \textbf{dominance threshold}: a single number per asset, independent of its monetary value, below which the quantum term dominates (Proposition~\ref{prop:dominance}).
\item A \textbf{scheduling justification}: ranking by cost of delay is optimal for sequencing work under constant loss rates, and the exact price of keeping classical and quantum work in separate queues is given (Propositions~\ref{prop:smith} and~\ref{prop:price}).
\item A \textbf{verification and statistical analysis} on four public systems, including input-uncertainty propagation, a signal-to-noise comparison, variance-based sensitivity indices, and an explicit report of where each ordering changes (Sections~\ref{sec:verification}--\ref{sec:stats}).
\end{enumerate}

\paragraph{What we do not claim.} We do not claim to be the first to score assets by post-quantum risk, to rank them, to gate on protection class or harvestability, or to observe that deferral under interception carries an accruing cost. These are prior work (Section~\ref{sec:related}). The propositions are standard probability and scheduling arguments applied to a new setting; the contribution is the construction and its verification, not new mathematics. The parameters of the case studies are assumed from public documentation. They are not elicited and not validated against outcomes.

\section{Related work}\label{sec:related}

\paragraph{Timing models.} Mosca's inequality \citep{mosca2018} frames migration as a feasibility constraint: with $X$ the years data must remain confidential, $Y$ the years migration takes and $Z$ the years to a CRQC, action is required when $X+Y>Z$. It is a binary rule on a point estimate. Baradziej \citep{baradziej2026} compares four model families for migration timing (Mosca, expected value, real options and risk prioritisation) and shows that their divergent recommendations trace to three choices: how CRQC arrival is represented, the discount rate, and whether interception harm is treated as present-tense. He observes that under interception deferral leaks, rather than preserves, the payoff. His article is a conceptual synthesis and proposes no per-asset model. Scholten et al.\ \citep{scholten2024} assess the benefits and risks of quantum computers more broadly. Recent engineering estimates lower the resources needed to factor a 2048-bit RSA modulus \citep{gidney2021,gidney2025}, and Blanco-Romero et al.\ argue that retaining intercepted traffic is economically cheap, so access to the channel, not storage, limits harvesting \citep{blanco2026}.

\paragraph{Per-asset scoring.} The Quantum-Adjusted Risk Score (QARS) of Grigaliūnas and Brūzgienė \citep{grigaliunas2025} is a weighted sum of a timeline term derived from Mosca's ratio, a sensitivity term and an exposure term gated by cryptographic vulnerability and harvestability, scored on $[0,1]$ per asset. It is quantum-only and dimensionless. The authors state the aim of embedding it in FAIR or NIST~800-30 registers ``alongside classical threats rather than siloed'' but give no method. Its 2026 sequel \citep{krusniauskas2026} makes the framework layer-specific, scoring data in transit, in use and at rest separately, adds an attenuation for symmetric protection bounded by Grover's algorithm, and reports an aggregate and a maximum-layer score on $[0,1]$; it remains dimensionless and quantum-only. Ishchukova et al.\ \citep{ishchukova2026} multiply CVSS severity by a quantum factor and an HNDL score to triage 78{,}587 CVEs, which ranks vulnerabilities, not assets, and is again dimensionless. (We read both 2026 papers at abstract level only: the publisher blocks automated retrieval.) CARAF \citep{ma2021} scores migration risk as a product of timeline and cost factors. Erlemann et al.\ attribute post-quantum exposure over an asset knowledge graph \citep{erlemann2026}. QUASAR \citep{weinberg2025} and the threat-modelling work of Baseri et al.\ \citep{baseri2024,baseri2025} address quantum risk in architecture and STRIDE-style assessment; the survey of Baseri et al.\ \citep{baseri2025}, which we read in full (version 4), combines STRIDE with NIST SP~800-30 in a qualitative Low/Medium/High matrix and assesses the pre-transition, hybrid and post-quantum phases separately, with no common unit. Two patents describe scoring quantum risk per enterprise asset from metadata \citep{us12519809} and valuing a networked device from the quantum safety of its ciphers \citep{us12010210}.

\paragraph{Classical leg and cost of delay.} Annualised loss expectancy is standard in quantitative risk analysis \citep{freund2015,hubbard2016} and the economics of security investment is developed by Gordon and Loeb \citep{gordon2002}. Cost of delay is established in product-development economics \citep{reinertsen2009}, and sequencing by weighted completion time is classical scheduling theory \citep{smith1956,pinedo2012}. Competing risks are textbook \citep{prentice1978}; our contribution is their use to couple two security threats on one asset.

\paragraph{Positioning.} In the sources we examined, no work places classical and post-quantum confidentiality risk on one asset in a common currency-rate unit and couples the quantum term to that asset's classical hazard. This is a statement about what we found. Our search was a literature and web search, not a systematic review, and we did not search patent classes exhaustively.

\paragraph{Related work by the author.} Earlier work by the author addresses the \emph{inputs} to this prioritisation rather than the prioritisation itself: LLM-assisted static analysis that finds quantum-vulnerable cryptography in code and scores it \citep{shaw2026a}, and a hybrid-by-default cryptographic library that implements the migration \citep{shaw2026b}. Where an asset's protection class, supplied here by hand, would in practice come from such an inventory, that earlier work is the natural source. A third paper concerns manifest and infrastructure-as-code analysis \citep{shaw2026c} and a fourth the validation of an evaluation harness \citep{shaw2026d}; the latter informs the verification discipline used here (Section~\ref{sec:verification}) but is otherwise unrelated.

\section{Model}\label{sec:model}

\subsection{Notation and assumptions}

Let $V$ be the loss (USD) if one retention period of an asset's data is disclosed, $\lambda$ the classical disclosure hazard (per year), $L$ the confidentiality lifetime (years), $h\in[0,1]$ the share of the data flow that a passive adversary can record, $\tau$ the annual turnover (new value per year divided by value at risk), $r=\tau V$ the value that newly flows each year (USD per year) and $\rho\ge0$ a discount rate; write $\kappa=\lambda+\rho$.

\begin{itemize}
\item[\textbf{A1}] The CRQC arrival time $T_q$ has a continuous distribution on $(0,\infty)$ with CDF $F$, density $f$, $F(0)=0$ and $F(t)>0$ for $t>0$. We use a Weibull law.
\item[\textbf{A2}] Data created at time $0$ is disclosed by a classical channel at time $T_c\sim\mathrm{Exp}(\lambda)$.
\item[\textbf{A3}] $T_c$ and $T_q$ are independent.
\item[\textbf{A4}] Data is sensitive during $[0,L]$ and worthless afterwards; disclosure at $t\le L$ by any channel loses its full value.
\item[\textbf{A5}] A fraction $h$ of the new value crosses a channel that can be recorded; value that cannot be recorded is not exposed to quantum loss.
\end{itemize}

\begin{definition}[Cost of delay]
\[
\CoD_c=\lambda V,\qquad
\CoD_q=r\,h\,I(\kappa,L),\qquad
I(\kappa,L)=\int_0^L f(t)\,e^{-\kappa t}\,dt,\qquad
\CoD=\CoD_c+\CoD_q .
\]
The quantum term is set to exactly zero when the key exchange is already quantum-safe, when $h=0$, or when $L=0$ (the gates).
\end{definition}

$I(\lambda,L)$ is the probability that a CRQC exists within the data's life \emph{and} the classical hazard has not already disclosed it (Proposition~\ref{prop:race}); $r\,h\,I$ is the rate at which each further year of delay adds irreversible quantum loss.

\subsection{The arrival law}

We take $F(t)=1-\exp\{-(t/\alpha)^\beta\}$ and fit $(\alpha,\beta)$ in closed form to two anchor points from the 2025 Quantum Threat Timeline Report \citep{moscapiani2025}, in which 26 experts gave the probability of a CRQC within 10 and 15 years as 28--49\% and 51--70\%. The upper ends define an optimistic scenario ($\alpha=13.18$, $\beta=1.433$, median 10.2 years) and the lower ends a pessimistic one ($\alpha=17.90$, $\beta=1.912$, median 14.8 years). Values beyond 15 years are extrapolation. The Weibull is an interpolation of two points, not a claim about the arrival process.

\subsection{Propositions}

Proofs are short and given in full in the repository (\texttt{methodology/FORMAL\_METHOD.md}); each proposition is also checked by an automated test on random continuous inputs.

\begin{proposition}[Units and scale invariance]\label{prop:units}
Both legs are in USD per year. Rescaling all losses by $c>0$ multiplies every $\CoD$ by $c$, so the ranking is invariant to a common rescaling of money.
\end{proposition}
\emph{Proof.} Both $\CoD_c=\lambda V$ and $\CoD_q=\tau VhI$ are linear in $V$ and $I,h,\tau,\lambda$ do not involve money. \hfill$\square$

\begin{proposition}[The race; no double counting]\label{prop:race}
Let $E_q=\{T_q\le L,\,T_q<T_c\}$ and $E_c=\{T_c\le L,\,T_c<T_q\}$. Then $\Pr(E_q)=I(\lambda,L)$; $\Pr(E_c)=\int_0^L\lambda e^{-\lambda t}(1-F(t))\,dt$; the events are disjoint with
\[
\Pr(E_c)+\Pr(E_q)=1-e^{-\lambda L}\bigl(1-F(L)\bigr),
\]
and the naive sum $\Pr(T_c\le L)+\Pr(T_q\le L)$ over-counts the union by exactly $(1-e^{-\lambda L})F(L)$.
\end{proposition}
\emph{Proof.} By independence $\Pr(T_q\in dt,\,T_c>t)=f(t)e^{-\lambda t}dt$; integrate. Likewise for $E_c$. Ties have probability zero, so the events are disjoint, and their union is $\{\min(T_c,T_q)\le L\}$ whose complement is $\{T_c>L,T_q>L\}$. The over-count is the intersection probability $\Pr(T_c\le L)\Pr(T_q\le L)$. \hfill$\square$

\begin{proposition}[Integration by parts]\label{prop:ibp}
For $\kappa\ge0$, $I(\kappa,L)=e^{-\kappa L}F(L)+\kappa\int_0^L F(t)e^{-\kappa t}\,dt$.
\end{proposition}
\emph{Proof.} Stieltjes integration by parts with $F(0)=0$. \hfill$\square$

For $1<\beta<2$ the Weibull density behaves like $t^{\beta-1}$ at the origin, which is not smooth; Simpson's rule on $fe^{-\kappa t}$ converges at order about $1.4$ (measured) whereas the form above, whose integrand behaves like $t^\beta$, converges at about $2.4$, with worst relative error $4.5\times10^{-10}$ against $9.4\times10^{-6}$.

\begin{proposition}[Bounds and coupling factor]\label{prop:bounds}
$e^{-\kappa L}F(L)\le I(\kappa,L)\le F(L)$. The coupling factor $I/F(L)$ lies in $[e^{-\kappa L},1]$ and equals $1$ only if $\kappa=0$.
\end{proposition}

\begin{proposition}[Comparative statics]\label{prop:statics}
$\partial I/\partial\kappa=-\int_0^Ltf(t)e^{-\kappa t}dt<0$; $\partial I/\partial L=f(L)e^{-\kappa L}>0$ where $f(L)>0$; if $F_1\ge F_2$ pointwise then $I_1\ge I_2$; $\CoD_q$ is linear and increasing in $r$ and $h$.
\end{proposition}

\begin{proposition}[Gating and discounting]\label{prop:gating}
$\CoD_q=0$ if and only if $rhL=0$ or the key exchange is quantum-safe. Discounting at $\rho$ is identical to adding $\rho$ to the classical hazard: $\CoD_q(\lambda,\rho)=\CoD_q(\lambda+\rho,0)$.
\end{proposition}

\begin{proposition}[Dominance threshold]\label{prop:dominance}
Let $V,\tau,h,L>0$ and $\Phi(\lambda)=\tau h\,I(\lambda+\rho,L)-\lambda$. Then $\CoD_q>\CoD_c\iff\Phi(\lambda)>0$, which does not depend on $V$. There is a unique $\lambda^*\in(0,\tau hF(L))$ with $\Phi(\lambda^*)=0$, and the quantum term dominates exactly when $\lambda<\lambda^*$. Necessary: $\lambda<\tau hF(L)$. Sufficient: $\lambda<\tau h\,e^{-(\lambda+\rho)L}F(L)$.
\end{proposition}
\emph{Proof.} Divide by $V$. $\Phi(0)=\tau hI(\rho,L)>0$; by Proposition~\ref{prop:bounds}, $\Phi(\tau hF(L))<0$; by Proposition~\ref{prop:statics}, $\Phi'<0$. The bounds follow from Proposition~\ref{prop:bounds}. \hfill$\square$

The quantum-led class is therefore \{low classical hazard, long lifetime, quantum-vulnerable protection, recordable traffic, high turnover\}, with a single boundary per asset.

\begin{proposition}[Not a weighted sum]\label{prop:notsum}
$\partial^2\CoD/\partial\lambda\,\partial L=-rhLf(L)e^{-\kappa L}<0$ whenever $rhLf(L)>0$. Hence $\CoD$ is not of the form $A(V,\lambda)+B(L,h,r,F,\rho)$, and in particular not $w_1S_c+w_2S_q$ for a classical score $S_c$ and quantum score $S_q$ of their own inputs.
\end{proposition}
\emph{Proof.} $\partial_\lambda\CoD=V-rh\int_0^Ltf e^{-\kappa t}dt$; differentiate in $L$. Any additive separable form has zero mixed partial. \hfill$\square$

This shows that the coupling is real; it does not show that the coupled form is more accurate. A competing-risks term is textbook \citep{prentice1978}, and the claim concerns its application.

\begin{proposition}[Classical leg is not materially double counted]\label{prop:overcount}
$\Pr(T_q<T_c\mid T_c\le1)\le F(1)$, which is $2.4\%$ (optimistic) and $0.4\%$ (pessimistic).
\end{proposition}

\begin{proposition}[Fitting the arrival law]\label{prop:fit}
For anchors $(t_1,p_1),(t_2,p_2)$ with $0<t_1<t_2$ and $0<p_1<p_2<1$ there is exactly one Weibull $(\alpha,\beta)$ through them: $\beta=\ln(y_2/y_1)/\ln(t_2/t_1)$, $\alpha=t_1y_1^{-1/\beta}$ with $y_i=-\ln(1-p_i)$. Two distinct fitted laws cross at most once, here at $t^*\approx44.7$ years where both CDFs exceed $0.997$.
\end{proposition}

\begin{proposition}[Smith's rule]\label{prop:smith}
A single team does items $j=1..n$ back to back; item $j$ takes $p_j$ years and costs a constant $w_j$ per year until finished. Total loss $\sum w_jC_j$ is minimised by non-increasing $w_j/p_j$; with equal durations, by highest cost of delay first.
\end{proposition}
\emph{Proof.} Adjacent interchange: swapping adjacent $i,j$ changes the loss by $w_jp_i-w_ip_j$, so sorting by $w/p$ never increases it \citep{smith1956}. \hfill$\square$

\begin{proposition}[Exact price of a sub-optimal order]\label{prop:price}
With $\rho_j=w_j/p_j$, any order exceeds Smith's loss by exactly $\sum p_ip_j|\rho_i-\rho_j|$ over pairs inverted relative to Smith's order. Two sequential queues (all classical fixes, then all post-quantum migrations) are optimal if and only if their $\rho$-ranges do not interleave.
\end{proposition}

Propositions~\ref{prop:smith} and~\ref{prop:price} assume constant rates. The quantum rate actually rises with time, so Smith's order is a first-order rule; Section~\ref{sec:sched} evaluates the exact time-varying accrual and finds the true optimum by dynamic programming.

\begin{figure}[ht]
\centering
\includegraphics[width=0.95\linewidth]{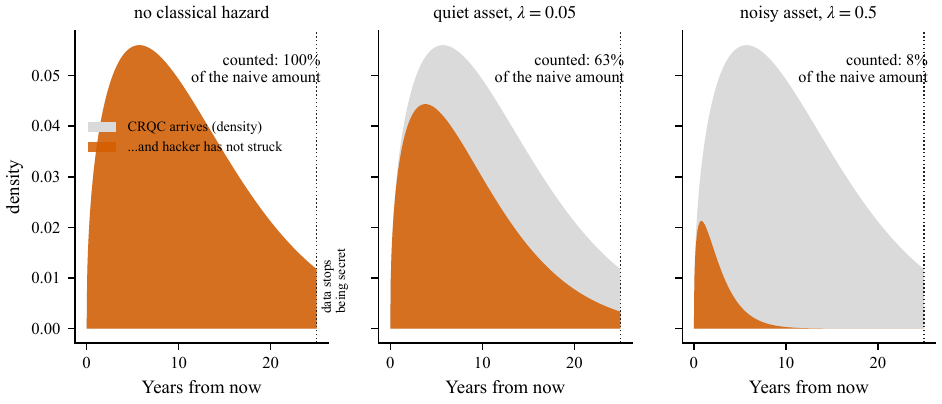}
\caption{The race between a classical attacker and a quantum computer for the same data (Proposition~\ref{prop:race}).}\label{fig:race}
\end{figure}

\begin{figure}[ht]
\centering
\includegraphics[width=0.95\linewidth]{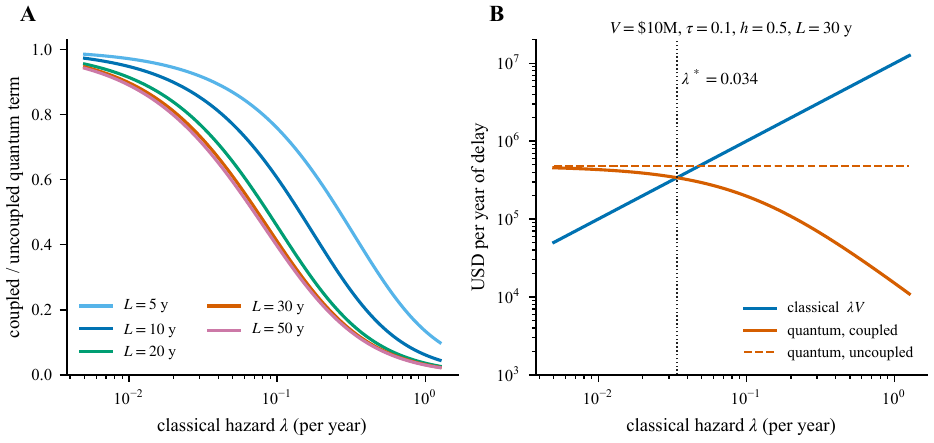}
\caption{The coupling factor, and the dominance threshold $\lambda^*$ below which the quantum term exceeds the classical term (Proposition~\ref{prop:dominance}).}\label{fig:dom}
\end{figure}

\section{Parameters: where the numbers come from}\label{sec:params}

All inputs are ordinal bands with fixed values; every case-study row records a rationale and a provenance. Table~\ref{tab:bands} gives the scales and their anchors. Only the \emph{relative} loss across assets affects the ranking (Proposition~\ref{prop:units}), so the absolute calibration matters for the stated magnitudes only.

\begin{table}[ht]
\centering\small
\begin{tabular}{p{2.6cm}p{3.6cm}p{7.6cm}}
\toprule
Input & Band values & Anchor and strength of evidence\\
\midrule
Loss $V$ & \$0, 10k, 100k, 1M, 10M, 100M & IBM/Ponemon 2025: global mean breach \$4.44M, about \$160 per record \citep{ibm2025}. Spacing is a judgement; ranking depends only on relative values.\\
Hazard $\lambda$ & 0.01, 0.05, 0.2, 0.5, 1.0 / yr & Ponemon 2018: 27.9\% recurrence within two years (about 0.16 / yr for an organisation) \citep{ponemon2018}. Per-store bands are judgement.\\
Lifetime $L$ & years, from a rule or fact & FATF R.11: at least five years \citep{fatf2012}; medical-retention rules (decades); biometrics capped at 50 years. The strongest input.\\
Recordable $h$ & 0, 0.1, 0.5, 0.9, 1.0 & Reasoning about network position; recording is cheap, access is the limit \citep{blanco2026}. Weak.\\
Turnover $\tau$ & per year, with a consistency rule & No published source for these systems. Rule: $V_{\text{flow}}\approx\tau_{\text{store}}V_{\text{store}}$. Weakest.\\
Protection class & observed or documented & NIST FIPS 203 and the behaviour of Shor's algorithm \citep{shor1997,fips203}.\\
\bottomrule
\end{tabular}
\caption{Input scales and their anchors. Full derivations, ranges and the revision history are in \texttt{methodology/PARAMETERS.md}.}\label{tab:bands}
\end{table}

\paragraph{A correction found by the consistency rule.} In an earlier draft two flow assets of the identity platform (registration packets and internal service traffic) were given stock-scale loss. Applying the turnover consistency rule showed this overstated them roughly tenfold; they were corrected to the HIGH band, all results were regenerated, and the headline asset (registration packets) remained first but the order became sensitive to the discount rate. This is reported because it shows how much the conclusions depend on a modelling choice that nobody publishes.

\section{Verification}\label{sec:verification}

Three independent layers support the implementation; none establishes that the model is empirically right.

\begin{enumerate}
\item \textbf{Mathematical.} Each proposition is proved above or in the repository, and checked by a test on random continuous inputs, with symbolic (SymPy) checks of the derivative and integration-by-parts identities.
\item \textbf{Numerical.} A Monte Carlo simulation of the two competing clocks agrees with the closed form on 46 asset-by-scenario cases within three standard errors; the integration-by-parts quadrature agrees with adaptive quadrature to a worst relative error of $4.5\times10^{-10}$.
\item \textbf{Second implementation.} A JavaScript port, used by the interactive explorer that accompanies this paper, reproduces 577 Python values to $10^{-7}$ relative.
\end{enumerate}

In addition, 256 automated tests run, including refutation baselines (ranking by data lifetime alone, by classical cost alone, by Mosca's binary rule). Reference entries were checked against Crossref and arXiv by script; 41 of 41 machine-verifiable entries resolved, and the remainder are reports, standards and patents checked by hand. Two guessed DOIs were wrong and were caught and corrected by that script. This verification discipline follows the author's earlier work on validated evaluation harnesses \citep{shaw2026d}.

\begin{figure}[ht]
\centering
\includegraphics[width=0.95\linewidth]{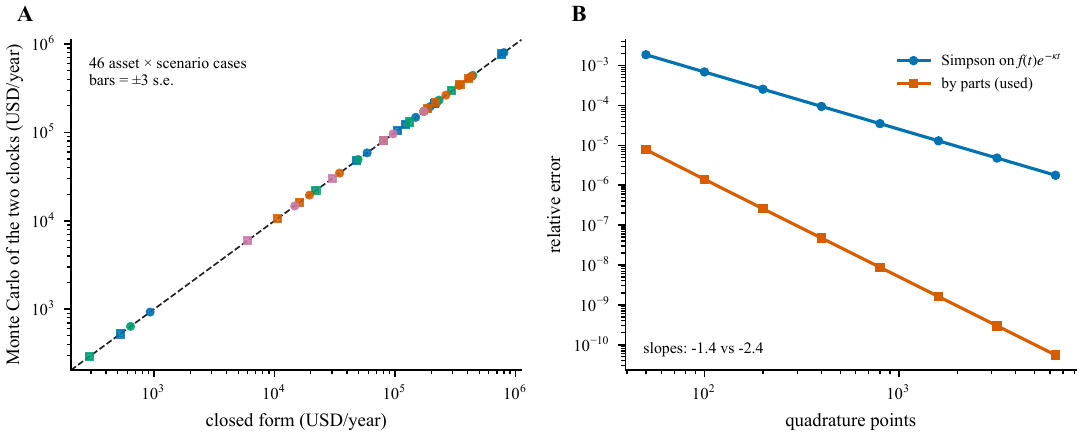}
\caption{Numerical validation: Monte Carlo against the closed form, and quadrature error of the two integration forms.}\label{fig:val}
\end{figure}

\section{Case studies}\label{sec:cases}

\subsection{What these are}
Four public systems: MOSIP (national identity), Apache Fineract (core banking), OpenMRS (medical records) and Google's Online Boutique (a demonstration web shop). All parameters are assumed from public documentation; \emph{no organisation supplied data}. Each asset is a category of data that the system's design contains; the rows, rationales and provenance are in \texttt{case\_studies/}. These are illustrations of what the method produces, not evidence that its inputs are right. Online Boutique is a mock system, included to test the procedure.

\begin{table}[ht]
\centering\small\setlength{\tabcolsep}{4pt}
\begin{tabular}{lrrrrcc}
\toprule
System & Assets & Median $L$ & Pairs differing & Quantum-led & Stable (CRQC) & Stable ($\rho$)\\
\midrule
MOSIP & 9 & 50 & 5 of 36 & 4 of 9 & yes & no\\
Apache Fineract & 9 & 10 & 0 of 36 & 1 of 9 & yes & yes\\
OpenMRS & 10 & 20 & 6 of 45 & 1 of 10 & no & no\\
Online Boutique & 7 & 4 & 2 of 21 & 0 of 7 & yes & yes\\
\bottomrule
\end{tabular}
\caption{Cross-system summary (optimistic arrival scenario). ``Pairs differing'' counts pairs of assets ordered differently from a classical-only ranking.}\label{tab:cross}
\end{table}

\begin{figure}[ht]
\centering
\includegraphics[width=0.95\linewidth]{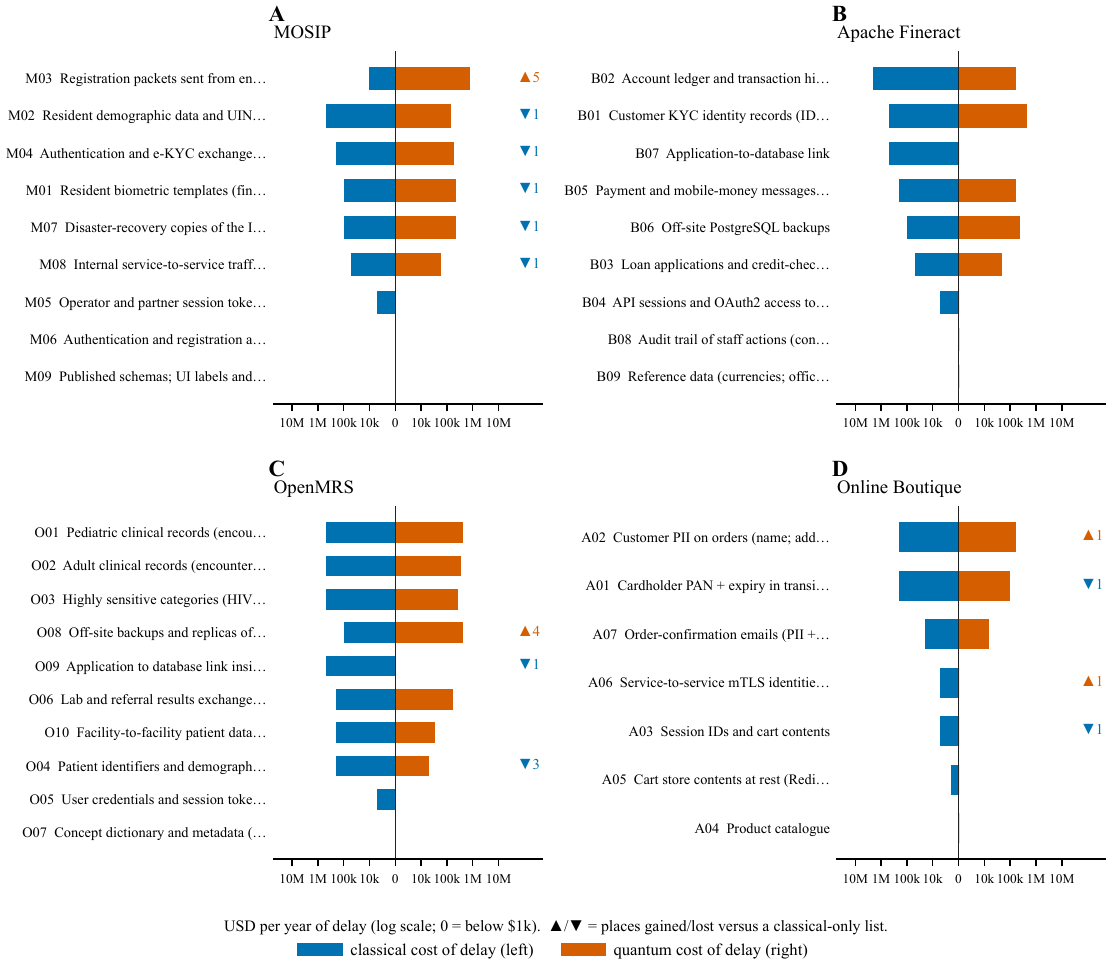}
\caption{Classical and quantum cost of delay per asset in each system (log scale).}\label{fig:cod}
\end{figure}

\paragraph{MOSIP.} Registration packets sent from enrolment centres rank first against sixth under a classical-only ranking: classical cost \$10.0k per year, quantum cost \$800.7k per year, because the packets are assumed to use an RSA-based key wrap, cross a network, are new every year and must stay secret for 50 years (the cap beyond which the arrival curve is wholly extrapolated). Four of nine assets are quantum-led. The order is invariant across arrival scenarios but changes with the discount rate: from $\rho\ge5\%$ the demographic store overtakes the packets.

\paragraph{Apache Fineract.} Records are kept for 5--15 years. The quantum term raises magnitudes (off-site backups: \$232.3k per year quantum against \$100.0k classical) but changes no pairwise decision: a negative result for medium-lifetime financial data.

\paragraph{OpenMRS.} Records of minors are kept for decades. Off-site backups, rarely attacked directly (classical \$100.0k per year) but replicated over partly exposed paths for 30 years, have a quantum cost of \$431.6k per year, 4.3 times their classical cost, and move from eighth to fourth place. The application-to-database link, with no quantum-vulnerable key exchange, has a quantum cost of exactly zero.

\paragraph{Online Boutique.} No asset is quantum-led. It is included only to test the procedure.

\begin{figure}[ht]
\centering
\includegraphics[width=0.95\linewidth]{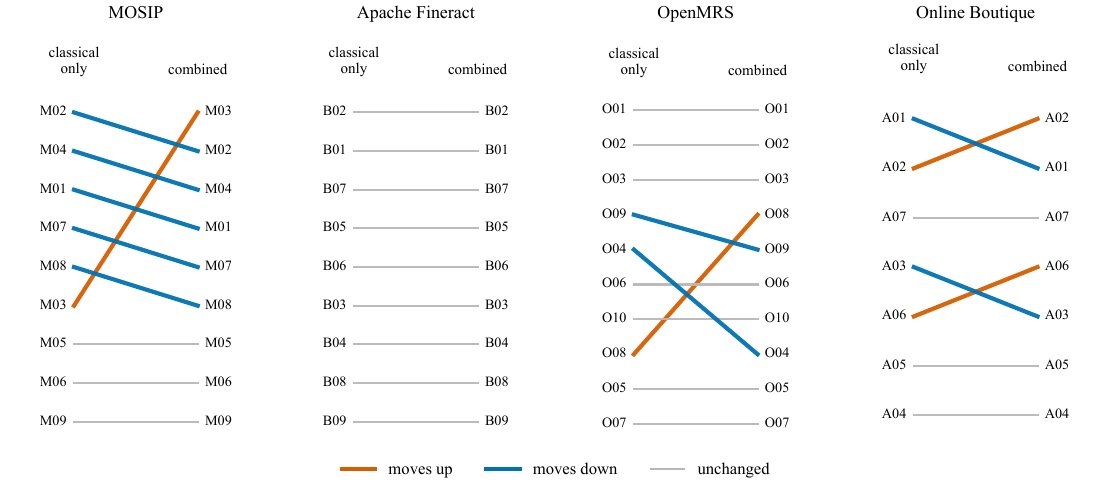}
\caption{Rank of each asset under a classical-only list against the unified list.}\label{fig:shift}
\end{figure}

\subsection{Where the ordering changes}
On OpenMRS the ordering is not invariant. Off-site backups outrank the application-to-database link if the CRQC median is before about 2039 or if $\rho<0.7\%$; otherwise the link outranks them. This boundary, not a single ordering, is the decision-relevant uncertainty (Figure~\ref{fig:sens}). A point-score method would hide it.

\begin{figure}[ht]
\centering
\includegraphics[width=0.95\linewidth]{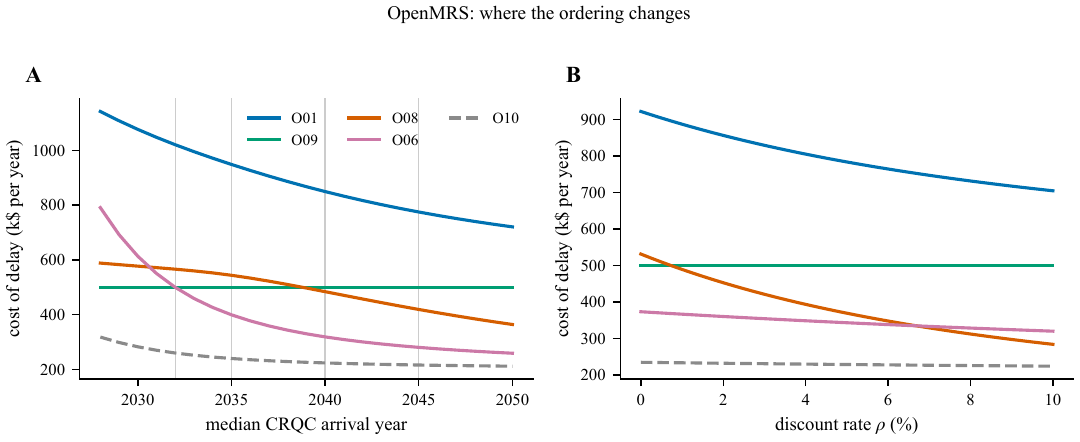}
\caption{Ordering of assets as the CRQC median year and the discount rate vary.}\label{fig:sens}
\end{figure}

\section{Statistical analysis}\label{sec:stats}

Every input is an ordinal band, so the central device is input-uncertainty propagation: each input of each asset is re-drawn from a stated range (value $\pm$ half a band; hazard $\times/\div2$; lifetime $\times/\div1.5$; recordable share $\pm0.1$; turnover $\times/\div3$; CRQC scenario optimistic or pessimistic with equal probability; discount uniform on $[0,7\%]$), the whole register is re-scored, and we count how often each conclusion survives. The ranges are themselves assumptions. All procedures use fixed seeds. Posited assets are excluded.

\subsection{Which conclusions survive}
Over 4,000 perturbed registers per system (Wilson intervals \citep{wilson1927}), the probability that the funded top-$k$ set differs from the classical-only set ($k=\mathrm{round}(0.25n)$, an arbitrary planning cut) is 71\% for MOSIP, 45\% for OpenMRS, 26\% for Fineract and 2\% for Online Boutique. Table~\ref{tab:surv} gives the key asset of each system.

\begin{table}[ht]
\centering\small
\begin{tabular}{p{4.6cm}rrrr}
\toprule
Asset & $P$(quantum-led) & Top-$k$, unified & Top-$k$, classical & $P$(up $\ge2$)\\
\midrule
MOSIP, registration packets & 100\% & 55\% & 0\% & 91\%\\
OpenMRS, off-site backups & 87\% & 19\% & 1\% & 56\%\\
Fineract, off-site backups & 64\% & 8\% & 2\% & 14\%\\
Boutique, customer PII & 26\% & 99\% & 99\% & 0\%\\
\bottomrule
\end{tabular}
\caption{Per-asset survival of conclusions over 4,000 perturbed registers.}\label{tab:surv}
\end{table}

\begin{figure}[ht]
\centering
\includegraphics[width=0.95\linewidth]{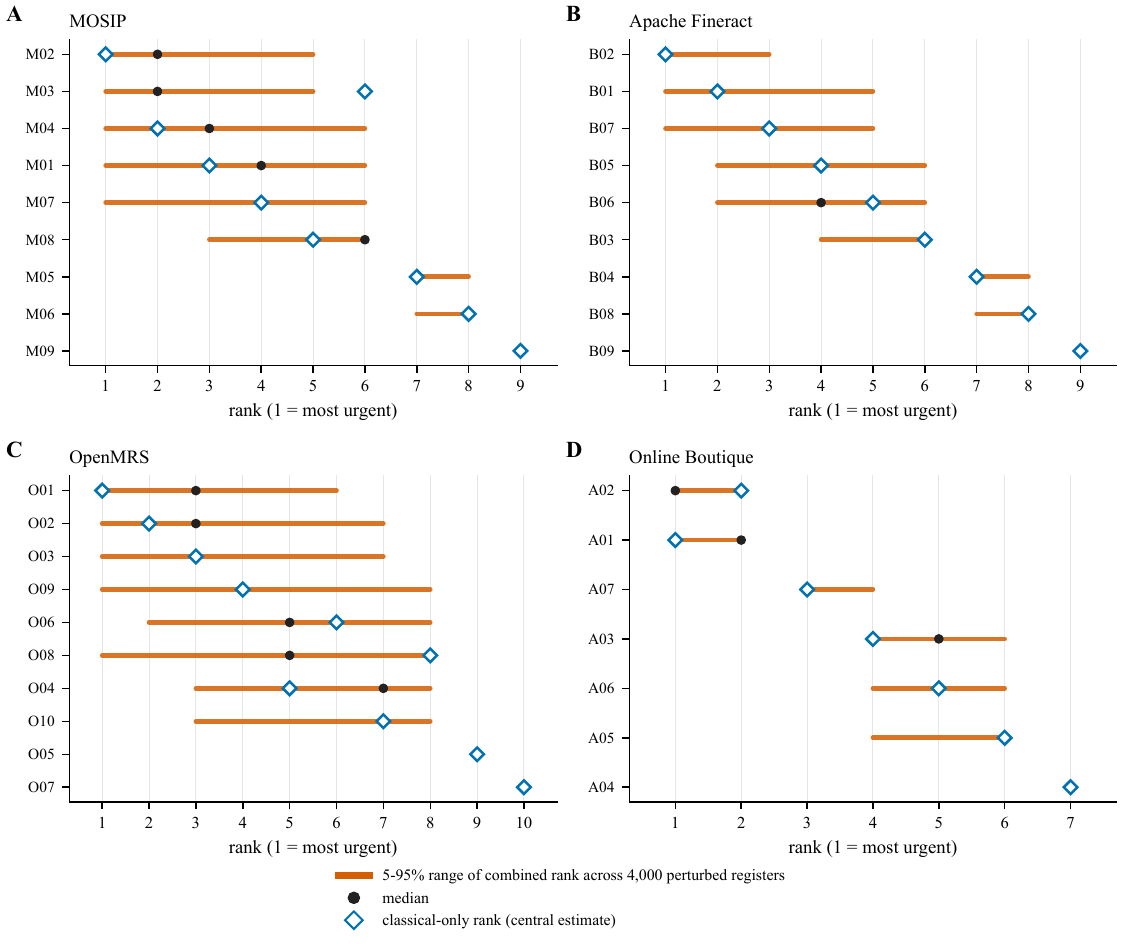}
\caption{Interval of ranks (5th to 95th percentile) for each asset across perturbed registers.}\label{fig:rankunc}
\end{figure}

\subsection{Does the quantum term reorder the list more than input noise does?}
The quantum term reorders some pairs, but the list also churns because two analysts would assign slightly different bands. On the same perturbed registers we compare \emph{signal} (pairs ordered differently by the combined and classical-only rankings) with \emph{noise} (pairs ordered differently between the classical-only rankings of two independent perturbations). The ratio exceeds one only for MOSIP: 1.36 (95\% interval 1.33--1.38), against 0.46 (Fineract), 0.41 (OpenMRS) and 0.14 (Online Boutique). For three of four systems the quantum term moves the list \emph{less} than analyst disagreement does. The claim that the quantum term reshuffles priorities is therefore not supported in general. What is supported is narrower: it changes magnitudes for long-lived assets and moves specific assets decisively, without reordering the whole list beyond noise for medium-lifetime systems.

\begin{figure}[ht]
\centering
\includegraphics[width=0.8\linewidth]{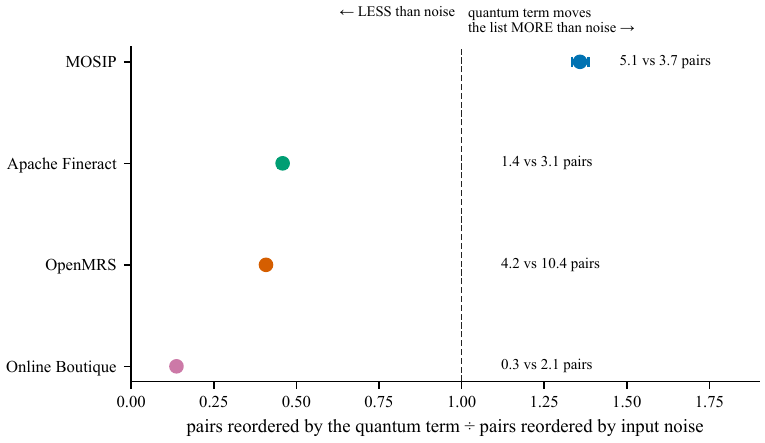}
\caption{Signal-to-noise ratio of the quantum reordering against input uncertainty, with 95\% bootstrap intervals.}\label{fig:snr}
\end{figure}

\subsection{What drives the answer}
We estimate Sobol' indices \citep{sobol2001,saltelli2010} with Saltelli sampling, Jansen estimators \citep{jansen1999} and a scrambled Sobol' sequence ($N=4096$, $36{,}864$ model runs per asset, 300 bootstrap resamples) on $\log_{10}\CoD_q$. First-order shares, and total-order within $0.02$ of them everywhere, show that the model is nearly additive in log space. Loss $V$ and turnover $\tau$ together explain 91\%, 79\%, 89\% and 73\% of the variance for the key assets of MOSIP, Fineract, OpenMRS and Online Boutique; the CRQC median year explains 3--17\%, the discount 1--5\%, and lifetime, recordable share and hazard under 10\% each. For assets with 30--50-year lifetimes the arrival date barely matters because $F(L)$ has saturated near one. Because $V$ and $\tau$ enter multiplicatively and their assumed ranges are similar, they share the variance by construction; the shares depend on the ranges chosen.

\begin{figure}[ht]
\centering
\includegraphics[width=0.9\linewidth]{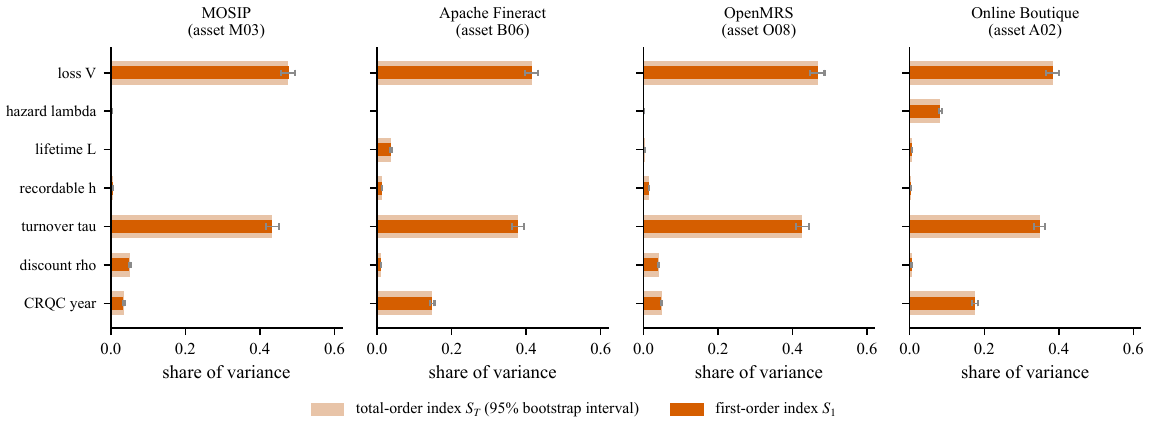}
\caption{First-order Sobol' indices for the quantum cost of delay of the key asset of each system.}\label{fig:sobol}
\end{figure}

\subsection{Lifetime and quantum share}
Across all 35 assets the Spearman correlation between lifetime and quantum share is $\rho_s=0.71$ (cluster-bootstrap 95\% interval $[0.47,0.89]$ over four systems; within-system permutation $p=10^{-4}$). This association is largely built into the model (Propositions~\ref{prop:statics} and~\ref{prop:gating}); it confirms that the implementation behaves as the theory predicts and that the assets were assigned lifetimes that expose the effect. It is not an empirical discovery about real systems.

\subsection{Monte Carlo convergence}
For OpenMRS off-site backups the probability of being quantum-led settles at $0.878$ with a Wilson half-width of $0.005$ at 16,000 draws ($0.010$ at 4,000), falling as $1/\sqrt n$.

\section{Does one list beat two? Scheduling}\label{sec:sched}

We build a backlog with one classical and one post-quantum item per asset, each taking one team-year; the quantum item accrues loss at its exact time-varying rate. We compare one list ranked by cost of delay (Smith's order), two queues (classical first or quantum first) and the exact optimum by dynamic programming over subsets.

\begin{table}[ht]
\centering\small\setlength{\tabcolsep}{4pt}
\begin{tabular}{lrrrrrr}
\toprule
System & Items & One list & Class.\ first & Quant.\ first & Optimum & Ratio\\
\midrule
MOSIP & 16 & 10.81 & 23.31 & 14.00 & 10.29 & 1.05\\
Apache Fineract & 15 & 16.93 & 25.93 & 32.75 & 15.97 & 1.06\\
OpenMRS & 17 & 30.75 & 40.00 & 36.84 & 28.96 & 1.06\\
Online Boutique & 11 & 2.49 & 5.67 & 3.43 & 2.27 & 1.10\\
\bottomrule
\end{tabular}
\caption{Committed loss (USD millions). One list ranked by cost of delay is within 5--10\% of the exact optimum; sequential queues cost 27--150\% more.}\label{tab:sched}
\end{table}

This is a \emph{model-internal} result: it says the ranking is optimal for the model's own loss definition, with assumed unit durations and a single team. It does not show that real organisations accrue loss this way.

\section{Comparison with QARS}

We implemented QARS from its published equations (equal weights, the paper's linear timeline surrogate); the implementation reproduces the paper's Table~2 values (0.733 and 0.383). This is our reading of the published model, not the authors' code. Kendall's $\tau$ \citep{kendall1938} between QARS and cost of delay is 0.83 (MOSIP), 0.67 (Fineract), 0.64 (OpenMRS) and 0.71 (Online Boutique). Three differences follow from the equations, each pinned by a unit test: QARS's gate zeroes only the exposure term, so an asset already on ML-KEM can still score 0.67; QARS has no classical-hazard input, so assets differing only in hazard score identically; and QARS maps value to a four-label scale, so a tenfold change inside a band leaves its score unchanged. This shows cost of delay behaves more sensibly on cases that QARS's equations handle poorly; it does not show cost of delay is more accurate. QARS was designed as a quantum-only score, and neither has been validated against outcomes.

\section{Threats to validity}

\begin{enumerate}
\item \textbf{Parameters are assumed, not elicited.} The most important limitation. Two practitioners may disagree and inter-rater agreement is unmeasured; one analyst, assisted by a language model, populated every register.
\item \textbf{Turnover and recordable share are the weakest inputs}, and turnover sets the size of the quantum-led class; no elicitation protocol was validated.
\item \textbf{Independence and constant hazard} (A2, A3) are assumptions. A state actor with both capabilities violates independence.
\item \textbf{Confidentiality only.} Systems whose dominant risks are integrity or availability will score low everywhere.
\item \textbf{Grover's algorithm is not modelled}: adequate for AES-256, not for AES-128. (The QARS sequel \citep{krusniauskas2026} does model a Grover-bounded attenuation.)
\item \textbf{The arrival law is an interpolation} of two published points, extrapolated beyond 15 years where most long-lived assets sit.
\item \textbf{Optimality is model-internal}, and Smith's rule is exact only for constant rates and one team.
\item \textbf{No outcome validation.} No dataset links these scores to realised harm.
\item \textbf{Perturbation ranges are assumptions}, and the signal-to-noise ratio and Sobol' shares depend on them.
\item \textbf{Prior-art coverage is incomplete}: we did not run a systematic review, and two closely related 2026 papers \citep{krusniauskas2026,ishchukova2026} were read at abstract level only.
\item \textbf{The systems are public software, not deployments.} We assessed design documentation, not any operator's configuration.
\end{enumerate}

\section{Discussion}

The contribution is narrower than a new threat-modelling methodology. It is a way to place two kinds of confidentiality risk in one unit, with a coupling that prevents double counting, a dominance condition that says which assets the quantum term can change, and an explicit report of where the answer depends on assumptions nobody can resolve today. In the case studies a planner would find that the quantum term matters decisively for a particular class (long-lived, quiet, recordable, quantum-vulnerable), moves most lists less than analyst disagreement would, and changes the order for the medical-record system only under particular arrival dates or discount rates.

\paragraph{Future work and collaboration.} The next step is to populate the same register format with an organisation's own parameters and measure inter-rater agreement; the register schema, the elicitation questions and the CLI are in the repository for that purpose. The author would welcome contact from operators of any of the systems studied, or of comparable systems, who wish to test the method on real inventories. Other extensions include a per-adversary exposure in place of the single recordable share, a non-constant classical hazard, and linking the inputs to automated cryptographic inventories such as those produced by code-analysis tools \citep{shaw2026a,cyclonedxcbom}.

\section{Conclusion}

We expressed classical and harvest-now-decrypt-later risk as a cost of delay in currency per year on the same asset, coupled through a competing-risks term. The construction has twelve proved propositions, is verified numerically and by a second implementation, and is applied to four public systems under assumed parameters. The result is conditional: the quantum term changes the order beyond input noise only where data must stay secret for decades, it moves particular assets decisively, and the arrival date matters far less than how much value flows and how much of it is new each year.

\section*{Declarations}

\paragraph{Use of AI tools.} A large language model (Anthropic's Claude, accessed through the Claude Code tool) assisted with writing code, tests, figures, the interactive explorer and the drafting of this manuscript, under the author's direction. The language model is not an author and takes no responsibility for the work. The author set the research question, chose the scope, decided which claims to retract after prior-art checks, and is responsible for the content. All numerical results are generated by code in the repository; every reference was checked against its source by script or by hand, and the repository README records how the tool was used.

\paragraph{Data and code availability.} The code, the four asset registers and the interactive explorer are available at \url{https://github.com/AnimeshShaw/pqc-cost-of-delay}; a live version of the explorer is at \url{https://animeshshaw.github.io/pqc-cost-of-delay/}. A permanent archive of release v0.1.0 is deposited at Zenodo (\url{https://doi.org/10.5281/zenodo.23133468}). Code and case-study data are released under the Apache License 2.0.

\paragraph{Funding.} None. \quad \textbf{Competing interests.} The author declares no competing interests.

\paragraph{Ethics.} The study uses only public documentation of open-source systems. No personal data and no human participants are involved.

\bibliographystyle{unsrtnat}
\bibliography{references/references}

\begin{thebibliography}{40}
\providecommand{\natexlab}[1]{#1}
\providecommand{\url}[1]{\texttt{#1}}
\expandafter\ifx\csname urlstyle\endcsname\relax
  \providecommand{\doi}[1]{doi: #1}\else
  \providecommand{\doi}{doi: \begingroup \urlstyle{rm}\Url}\fi

\bibitem[Mosca(2018)]{mosca2018}
Michele Mosca.
\newblock {Cybersecurity in an Era with Quantum Computers: Will We Be Ready?}
\newblock \emph{IEEE Security \& Privacy}, 16\penalty0 (5):\penalty0 38--41,
  2018.
\newblock \doi{10.1109/MSP.2018.3761723}.

\bibitem[{National Institute of Standards and
  Technology}(2024{\natexlab{a}})]{fips203}
{National Institute of Standards and Technology}.
\newblock {Module-Lattice-Based Key-Encapsulation Mechanism Standard}.
\newblock Technical Report FIPS 203, NIST, 2024{\natexlab{a}}.

\bibitem[{National Institute of Standards and
  Technology}(2024{\natexlab{b}})]{nistir8547}
{National Institute of Standards and Technology}.
\newblock {Transition to Post-Quantum Cryptography Standards}.
\newblock Technical Report NIST IR 8547 (initial public draft), NIST,
  2024{\natexlab{b}}.
\newblock Initial public draft, 12 November 2024: quantum-vulnerable public-key
  algorithms proposed to be deprecated after 2030 and disallowed after 2035;
  secondary.

\bibitem[Baradziej(2026)]{baradziej2026}
Simon Baradziej.
\newblock {The cost of waiting: a decision-theoretic synthesis of early versus
  late post-quantum migration under uncertainty}.
\newblock \emph{Frontiers in Quantum Science and Technology}, 5, 2026.
\newblock \doi{10.3389/frqst.2026.1918786}.

\bibitem[Scholten et~al.(2024)Scholten, Williams, Moody, Mosca, Hurley, Zeng,
  Troyer, and Gambetta]{scholten2024}
Travis~L. Scholten, Carl~J. Williams, Dustin Moody, Michele Mosca, William
  Hurley, William~J. Zeng, Matthias Troyer, and Jay~M. Gambetta.
\newblock {Assessing the Benefits and Risks of Quantum Computers}, 2024.
\newblock URL \url{https://arxiv.org/abs/2401.16317}.

\bibitem[Gidney and Ekerå(2021)]{gidney2021}
Craig Gidney and Martin Ekerå.
\newblock {How to factor 2048 bit RSA integers in 8 hours using 20 million
  noisy qubits}.
\newblock \emph{Quantum}, 5:\penalty0 433, 2021.
\newblock \doi{10.22331/q-2021-04-15-433}.

\bibitem[Gidney(2025)]{gidney2025}
Craig Gidney.
\newblock {How to factor 2048 bit RSA integers with less than a million noisy
  qubits}, 2025.
\newblock URL \url{https://arxiv.org/abs/2505.15917}.

\bibitem[Blanco-Romero et~al.(2026)Blanco-Romero, Mendoza, Rubio, Campo, and
  Sánchez]{blanco2026}
Javier Blanco-Romero, Florina~Almenares Mendoza, Carlos~García Rubio, Celeste
  Campo, and Daniel~Díaz Sánchez.
\newblock {On the Practical Feasibility of Harvest-Now, Decrypt-Later Attacks},
  2026.
\newblock URL \url{https://arxiv.org/abs/2603.01091}.

\bibitem[Šarūnas Grigaliūnas and Brūzgienė(2025)]{grigaliunas2025}
Šarūnas Grigaliūnas and Rasa Brūzgienė.
\newblock {Towards a Unified Quantum Risk Assessment}.
\newblock \emph{Electronics}, 14\penalty0 (17):\penalty0 3338, 2025.
\newblock \doi{10.3390/electronics14173338}.

\bibitem[Kru{\v s}niauskas et~al.(2026)Kru{\v s}niauskas, Grigali{\=u}nas,
  Br{\=u}zgien{\.e}, and Cayir]{krusniauskas2026}
Simas Kru{\v s}niauskas, {\v S}ar{\=u}nas Grigali{\=u}nas, Rasa
  Br{\=u}zgien{\.e}, and Mert Cayir.
\newblock {A Quantum-Adjusted Risk Model for Enterprise Infrastructure Across
  Data In Transit, In Use, and At Rest}.
\newblock \emph{Electronics}, 15\penalty0 (12):\penalty0 2546, 2026.
\newblock \doi{10.3390/electronics15122546}.

\bibitem[Ishchukova et~al.(2026)Ishchukova, Sajadi, Petrenko, Petrenko, and
  Nekrasov]{ishchukova2026}
Evgeniya Ishchukova, Faezeh~Sadat Sajadi, Sergei Petrenko, Alexey Petrenko, and
  Alexey Nekrasov.
\newblock {A Quantum Risk Index for Cryptographic CVEs: Empirical Evidence from
  the National Vulnerability Database, 2016--2026}.
\newblock \emph{Computers}, 15\penalty0 (8):\penalty0 495, 2026.
\newblock \doi{10.3390/computers15080495}.

\bibitem[Ma et~al.(2021)Ma, Colon, Dera, Rashidi, and Garg]{ma2021}
Chujiao Ma, Luis Colon, Joe Dera, Bahman Rashidi, and Vaibhav Garg.
\newblock {CARAF: Crypto Agility Risk Assessment Framework}.
\newblock \emph{Journal of Cybersecurity}, 7\penalty0 (1), 2021.
\newblock \doi{10.1093/cybsec/tyab013}.

\bibitem[Erlemann et~al.(2026)Erlemann, Morris, and Sathe]{erlemann2026}
Rasmus Erlemann, Charles~Colyer Morris, and Sanjyot Sathe.
\newblock {Full-Stack Knowledge Graph and LLM Framework for Post-Quantum Cyber
  Readiness}, 2026.
\newblock URL \url{https://arxiv.org/abs/2601.03504}.

\bibitem[Weinberg(2025)]{weinberg2025}
Abraham~Itzhak Weinberg.
\newblock {Preparing for the Post Quantum Era: Quantum Ready Architecture for
  Security and Risk Management (QUASAR) -- A Strategic Framework for
  Cybersecurity}, 2025.
\newblock URL \url{https://arxiv.org/abs/2505.17034}.

\bibitem[Baseri et~al.(2025{\natexlab{a}})Baseri, Chouhan, Ghorbani, and
  Chow]{baseri2024}
Yaser Baseri, Vikas Chouhan, Ali Ghorbani, and Aaron Chow.
\newblock {Evaluation framework for quantum security risk assessment: A
  comprehensive strategy for quantum-safe transition}.
\newblock \emph{Computers \& Security}, 150:\penalty0 104272,
  2025{\natexlab{a}}.
\newblock \doi{10.1016/j.cose.2024.104272}.

\bibitem[Baseri et~al.(2025{\natexlab{b}})Baseri, Hafid, and
  Lashkari]{baseri2025}
Yaser Baseri, Abdelhakim Hafid, and Arash~Habibi Lashkari.
\newblock {Future-Proofing Cloud Security Against Quantum Attacks: Risk,
  Transition, and Mitigation Strategies}, 2025{\natexlab{b}}.
\newblock URL \url{https://arxiv.org/abs/2509.15653}.

\bibitem[Miles et~al.(2026)Miles, Wood, and Ramos]{us12519809}
Brian~M. Miles, Elizabeth Wood, and Anthony Ramos.
\newblock {Scalable assessment and prioritization of network asset cryptography
  for quantum risk mitigation}.
\newblock US Patent 12,519,809 B2, AT\&T Intellectual Property I LP, 2026.
\newblock Granted 6 January 2026; claim 1 read via Google Patents; secondary.

\bibitem[McNutt and dos Santos(2024)]{us12010210}
Robert McNutt and Daniel~Ricardo dos Santos.
\newblock {Determining cryptographic quantum-safety for network assets}.
\newblock US Patent 12,010,210 B1, Forescout Technologies Inc., 2024.
\newblock Granted 11 June 2024; claim 1 read via Google Patents; secondary.

\bibitem[Freund and Jones(2015)]{freund2015}
Jack Freund and Jack Jones.
\newblock \emph{{Measuring and Managing Information Risk: A FAIR Approach}}.
\newblock Butterworth-Heinemann, 2015.
\newblock Year from the publisher's chapter records (Elsevier, 2015);
  secondary.

\bibitem[Hubbard and Seiersen(2016)]{hubbard2016}
Douglas~W. Hubbard and Richard Seiersen.
\newblock \emph{{How to Measure Anything in Cybersecurity Risk}}.
\newblock Wiley, 2016.
\newblock \doi{10.1002/9781119162315}.

\bibitem[Gordon and Loeb(2002)]{gordon2002}
Lawrence~A. Gordon and Martin~P. Loeb.
\newblock {The economics of information security investment}.
\newblock \emph{ACM Transactions on Information and System Security},
  5\penalty0 (4):\penalty0 438--457, 2002.
\newblock \doi{10.1145/581271.581274}.

\bibitem[Reinertsen(2009)]{reinertsen2009}
Donald~G. Reinertsen.
\newblock \emph{{The Principles of Product Development Flow: Second Generation
  Lean Product Development}}.
\newblock Celeritas Publishing, 2009.
\newblock Publisher and year confirmed from catalogue records; secondary.

\bibitem[Smith(1956)]{smith1956}
Wayne~E. Smith.
\newblock {Various optimizers for single‐stage production}.
\newblock \emph{Naval Research Logistics Quarterly}, 3\penalty0 (1-2):\penalty0
  59--66, 1956.
\newblock \doi{10.1002/nav.3800030106}.

\bibitem[Pinedo(2012)]{pinedo2012}
Michael~L. Pinedo.
\newblock \emph{{Scheduling: Theory, Algorithms, and Systems}}.
\newblock Springer, 4 edition, 2012.
\newblock \doi{10.1007/978-1-4614-2361-4}.

\bibitem[Prentice et~al.(1978)Prentice, Kalbfleisch, Peterson, Flournoy,
  Farewell, and Breslow]{prentice1978}
R.~L. Prentice, J.~D. Kalbfleisch, A.~V. Peterson, N.~Flournoy, V.~T. Farewell,
  and N.~E. Breslow.
\newblock {The Analysis of Failure Times in the Presence of Competing Risks}.
\newblock \emph{Biometrics}, 34\penalty0 (4):\penalty0 541, 1978.
\newblock \doi{10.2307/2530374}.

\bibitem[Shaw(2026{\natexlab{a}})]{shaw2026a}
Animesh Shaw.
\newblock {Quantum-Safe Code Auditing: LLM-Assisted Static Analysis and
  Quantum-Aware Risk Scoring for Post-Quantum Cryptography Migration},
  2026{\natexlab{a}}.
\newblock URL \url{https://arxiv.org/abs/2604.00560}.

\bibitem[Shaw(2026{\natexlab{b}})]{shaw2026b}
Animesh Shaw.
\newblock {quantum-safe: Bridging the Post-Quantum Production Gap with a
  Hybrid-by-Default Python Cryptography Library}, 2026{\natexlab{b}}.
\newblock URL \url{https://arxiv.org/abs/2605.17061}.

\bibitem[Shaw(2026{\natexlab{c}})]{shaw2026c}
Animesh Shaw.
\newblock {Compared to What? A Human-Anchored Security Benchmark for
  LLM-Generated Infrastructure-as-Code}, 2026{\natexlab{c}}.
\newblock URL \url{https://arxiv.org/abs/2608.28021}.

\bibitem[Shaw(2026{\natexlab{d}})]{shaw2026d}
Animesh Shaw.
\newblock {Silent Failures in Agentic Security Evaluation: A Validated Harness
  for Tool-Call Mediation Under Indirect Prompt Injection}, 2026{\natexlab{d}}.
\newblock URL \url{https://arxiv.org/abs/2609.32691}.

\bibitem[Mosca and Piani(2026)]{moscapiani2025}
Michele Mosca and Marco Piani.
\newblock {Quantum Threat Timeline Report 2025}.
\newblock Technical report, evolutionQ Inc. and Global Risk Institute, 2026.
\newblock URL
  \url{https://globalriskinstitute.org/publication/quantum-threat-timeline-report-2025b/}.
\newblock Published 9 March 2026; 26 experts surveyed. Figures confirmed from
  the GRI landing page and a second secondary source; full report not read.

\bibitem[{IBM Security} and {Ponemon Institute}(2025)]{ibm2025}
{IBM Security} and {Ponemon Institute}.
\newblock {Cost of a Data Breach Report 2025: The AI Oversight Gap}.
\newblock Technical report, IBM, 2025.
\newblock Read directly (pp. 1--15): global average USD 4.44M; healthcare
  7.42M; financial 5.56M; United States 10.22M; India 2.51M; USD 160 per
  customer-PII record; 600 organisations; 2,960--113,620 records.

\bibitem[{Ponemon Institute} and {IBM Security}(2018)]{ponemon2018}
{Ponemon Institute} and {IBM Security}.
\newblock {2018 Cost of Data Breach Study: Global Overview}.
\newblock Technical report, IBM, 2018.
\newblock URL \url{https://www.ibm.com/support/pages/node/6213647}.
\newblock The 27.9\% two-year recurrence likelihood and the 477-company sample
  are quoted from IBM's published summary page; secondary. Companies sampled
  had already experienced a breach.

\bibitem[{Financial Action Task Force}(2012)]{fatf2012}
{Financial Action Task Force}.
\newblock {International Standards on Combating Money Laundering and the
  Financing of Terrorism and Proliferation: The FATF Recommendations}.
\newblock Technical report, FATF, 2012.
\newblock As updated; Recommendation 11 (record-keeping) quoted from the
  published text found by search; secondary.

\bibitem[Shor(1997)]{shor1997}
Peter~W. Shor.
\newblock {Polynomial-Time Algorithms for Prime Factorization and Discrete
  Logarithms on a Quantum Computer}.
\newblock \emph{SIAM Journal on Computing}, 26\penalty0 (5):\penalty0
  1484--1509, 1997.
\newblock \doi{10.1137/S0097539795293172}.

\bibitem[Wilson(1927)]{wilson1927}
Edwin~B. Wilson.
\newblock {Probable Inference, the Law of Succession, and Statistical
  Inference}.
\newblock \emph{Journal of the American Statistical Association}, 22\penalty0
  (158):\penalty0 209--212, 1927.
\newblock \doi{10.1080/01621459.1927.10502953}.

\bibitem[Sobol'(2001)]{sobol2001}
I.M Sobol'.
\newblock {Global sensitivity indices for nonlinear mathematical models and
  their Monte Carlo estimates}.
\newblock \emph{Mathematics and Computers in Simulation}, 55\penalty0
  (1-3):\penalty0 271--280, 2001.
\newblock \doi{10.1016/S0378-4754(00)00270-6}.

\bibitem[Saltelli et~al.(2010)Saltelli, Annoni, Azzini, Campolongo, Ratto, and
  Tarantola]{saltelli2010}
Andrea Saltelli, Paola Annoni, Ivano Azzini, Francesca Campolongo, Marco Ratto,
  and Stefano Tarantola.
\newblock {Variance based sensitivity analysis of model output. Design and
  estimator for the total sensitivity index}.
\newblock \emph{Computer Physics Communications}, 181\penalty0 (2):\penalty0
  259--270, 2010.
\newblock \doi{10.1016/j.cpc.2009.09.018}.

\bibitem[Jansen(1999)]{jansen1999}
Michiel~J.W. Jansen.
\newblock {Analysis of variance designs for model output}.
\newblock \emph{Computer Physics Communications}, 117\penalty0 (1-2):\penalty0
  35--43, 1999.
\newblock \doi{10.1016/S0010-4655(98)00154-4}.

\bibitem[KENDALL(1938)]{kendall1938}
M.~G. KENDALL.
\newblock {A NEW MEASURE OF RANK CORRELATION}.
\newblock \emph{Biometrika}, 30\penalty0 (1-2):\penalty0 81--93, 1938.
\newblock \doi{10.1093/biomet/30.1-2.81}.

\bibitem[{OWASP CycloneDX}(2026)]{cyclonedxcbom}
{OWASP CycloneDX}.
\newblock {Cryptography Bill of Materials (CBOM)}, 2026.
\newblock URL \url{https://cyclonedx.org/capabilities/cbom/}.
\newblock Accessed 2026-10-02.

\end{thebibliography}

\end{document}